\documentclass[11pt]{article}
\usepackage[a4paper,margin=1in]{geometry}
\usepackage{authblk}
\usepackage{setspace}
\RequirePackage{amsthm,amsmath,amsfonts,amssymb}
\RequirePackage[authoryear]{natbib}
\RequirePackage{graphicx}
\RequirePackage{mathtools}
\RequirePackage{bm}
\RequirePackage{booktabs}
\RequirePackage{multirow}
\RequirePackage{enumitem}
\RequirePackage[hidelinks]{hyperref}

\theoremstyle{plain}
\newtheorem{theorem}{Theorem}
\newtheorem{proposition}{Proposition}
\newtheorem{lemma}{Lemma}
\newtheorem{corollary}{Corollary}
\theoremstyle{definition}
\newtheorem{definition}{Definition}
\newtheorem{assumption}{Assumption}
\newtheorem{remark}{Remark}

\DeclareMathOperator{\E}{\mathbb{E}}
\DeclareMathOperator{\Var}{Var}

\newcommand{\R}{\mathbb{R}}
\newcommand{\N}{\mathcal{N}}
\newcommand{\C}{\mathrm{C}^{+}}
\newcommand{\bx}{\bm{x}}
\newcommand{\bbeta}{\bm{\beta}}
\newcommand{\half}{\tfrac{1}{2}}
\newcommand{\ind}{\mathbf{1}}

\begin{document}

\title{Horseshoe Priors for Spatial Small Area Estimation:\
Regular Variation, Tail Robustness, and Deep Learning}

\author[1]{Dhiman Bhadra}
\author[2]{Nicholas G. Polson}
\affil[1]{Indian Institute of Management Ahmedabad}
\affil[2]{Booth School of Business, University of Chicago}
\date{}
\maketitle

\setlength{\parindent}{15pt}
\setlength{\parskip}{0pt}

\begin{abstract}
Small area estimation borrows strength across domains to repair the poor
precision of direct survey estimators. Two philosophies dominate the
area-level literature. The first, descending from \citet{ghoshrao1994} and
realized at scale in the spatial program of \citet{mercer2015} and
\citet{wakefield2019}, borrows strength through \emph{structured} Gaussian
smoothing: an intrinsic conditional autoregression or its BYM2
reparameterization pools each area towards its geographic neighbours. The
second borrows strength \emph{globally} but acts \emph{locally} through a
heavy-tailed global--local prior on the area effects, of which the horseshoe
of \citet{cps2010} is the canonical instance; \citet{tang2018} first brought
this idea to small area estimation. We study the horseshoe Fay--Herriot model
in the realistic regime of \emph{known and unequal} sampling variances and make
four contributions. First, a tail-robustness theorem: through a heteroscedastic
Tweedie identity the posterior mean leaves strongly signalled areas essentially
unshrunk, bounding the influence of an outlying direct estimate, in sharp
contrast to the linear, unbounded shrinkage of Gaussian random-effect models.
Second, standardizing by the known design variances transfers the minimax
contraction and credible-set theory of the homoscedastic sequence model to the
heteroscedastic Fay--Herriot problem; the posterior contracts at the
nearly-black minimax rate, with a matching lower bound that fixes the sharp
constant. Third, we give a precise account of when structured smoothing and
global--local shrinkage each win. Fourth, an $O(m)$ Gibbs sampler, a controlled
study, and an analysis of the Scottish lip cancer data confirm the account: on
that strongly spatial data the smoother predicts held-out districts best, yet the
horseshoe shrinkage coefficient flags the exceptional high- and low-incidence
districts that smoothing suppresses. We also connect the horseshoe to a
data-augmentation approach to Bayesian quantile regression and to deep quantile
networks for sharp spatial structure. Throughout we argue, following the
regular-variation theory of \citet{bdpw2016}, that these properties make the
horseshoe a sound \emph{default} prior for area effects: it borrows strength
aggressively yet lets genuinely exceptional areas speak for themselves, with no
tuning and no neighbourhood graph.
\end{abstract}


\section{Introduction}\label{sec:intro}

Demand for reliable statistics at fine geographic and demographic resolution has
far outstripped the sample sizes that surveys allocate to individual domains.
The \emph{direct} estimator for an area, computed from that area's own sample,
is approximately design-unbiased but so variable as to be useless when the
within-area sample is a handful of units or empty. Small area estimation (SAE)
repairs this by \emph{borrowing strength}: a model links the areas, and each
area's estimate is pulled towards a prediction formed from auxiliary
information and from the other areas. The authoritative appraisals of
\citet{ghoshrao1994}, \citet{rao2003} and \citet{ghosh2020} trace how synthetic
estimators, empirical best linear unbiased prediction (EBLUP), and empirical
and hierarchical Bayes methods came to organize the field, almost always
through the area-level model of \citet{fayherriot1979} or the unit-level model
of \citet{bhf1988}.

How strength is borrowed, and not merely that it is borrowed, is the substance
of a method. The area-level literature has settled into two broad philosophies,
which this paper places side by side and then separates theoretically.

\paragraph{Structured smoothing.}
The first philosophy treats the area effects as a smooth latent surface and
pools each area towards its neighbours. Its modern realization is the spatial
program developed in Seattle by Jon Wakefield and collaborators for global
health: \citet{mercer2015} marry design-based direct estimators with a
Bayesian space--time smoother for child mortality, \citet{wakefield2019}
develop the methodology for under-five mortality across the developing world,
and the \texttt{SUMMER} package \citep{summer} has made the approach a standard
tool for the Demographic and Health Surveys. The latent model is the
Besag--York--Molli\'e (BYM) prior \citep{besag1991}, an intrinsic conditional
autoregression (ICAR) plus an unstructured Gaussian term, most often used in
the \citet{riebler2016} BYM2 reparameterization that separates the total
random-effect variance from the fraction that is spatial, and fitted by INLA
\citep{rue2009}. Structured smoothing is powerful precisely when the truth is
spatially coherent: information flows along the adjacency graph, and a sparsely
sampled area inherits the level of its well-sampled neighbours.

\paragraph{Global--local shrinkage.}
The second philosophy makes no spatial assumption. It places a heavy-tailed
\emph{global--local} prior on the area effects: a single global scale pulls all
effects towards the regression surface, while heavy-tailed local scales let
individual areas escape that pull. The horseshoe prior \citep{cps2010,ps2011}
is the canonical instance, and \citet{tang2018} first cast it as a model for the
Fay--Herriot random effects, with multivariate and spatial extensions following
\citep{tangghosh2023}. Such priors borrow strength globally---most areas are
shrunk hard towards the synthetic prediction---yet act locally, leaving the few
genuinely deviant areas almost untouched. This is the natural model when
departures from the regression surface are \emph{sparse}: a minority of areas
behave anomalously and the rest are well explained by covariates. It also
furnishes a continuous answer to the question raised by \citet{datta2011} and
\citet{datta2015} of whether every area needs a random effect at all.

\paragraph{Contributions.}
Existing global--local work for small area estimation
\citep{tang2018,tangghosh2023} establishes methodology and posterior
consistency, but the practically decisive feature of the Fay--Herriot
setting---that the sampling variances $D_i$ are \emph{known and unequal}, being
design variances of the direct estimators---has not been turned to theoretical
advantage, nor has the horseshoe been squarely compared with the
structured-smoothing program it competes with in practice. Treating the known
$D_i$ as the lever rather than a nuisance, we make four contributions.
\begin{enumerate}[leftmargin=1.9em,itemsep=3pt,topsep=3pt]
\item \textbf{Bounded influence (Section~\ref{sec:robust}).} A heteroscedastic
Tweedie--Brown identity gives the posterior mean a \emph{redescending} score,
$\E(\theta_i\mid y_i)=y_i-2D_i/(y_i-\bx_i^\top\bbeta)+o(|y_i|^{-1})$, so a
strongly signalled area is returned essentially unshrunk and its
gross-error sensitivity is finite (Theorem~\ref{thm:robust},
Corollary~\ref{cor:influence}). Gaussian random-effect models---EBLUP and BYM2
alike---apply a fixed linear shrinkage whose sensitivity to a corrupted direct
estimate is unbounded. The property rests only on a regularly varying prior tail,
so it holds across the global--local family, not just the horseshoe
(Section~\ref{sec:family}).

\item \textbf{Sharp heteroscedastic minimaxity (Section~\ref{sec:minimax}).}
Standardizing by the known $D_i$ reduces the area-level model to a unit-variance
sequence model and transfers the contraction theory of \citet{vdp2014,vdp2017}
to the Fay--Herriot problem. Over nearly-black classes the posterior contracts at
the minimax rate (Theorem~\ref{thm:contract}); a matching lower bound fixes the
exact constant in the $D$-weighted loss (Proposition~\ref{prop:lower}); the
contraction is achieved adaptively, with no oracle knowledge of the number of
exceptional areas, by either an empirical-Bayes or a half-Cauchy global scale
(Corollary~\ref{cor:adapt}); and the credible balls are honest under the usual
excessive-bias conditions. The known design variances are exactly what make a
heavy-tailed prior provably minimax here.

\item \textbf{Separation of the two philosophies (Sections~\ref{sec:compare}--%
\ref{sec:df}).} We prove that structured smoothing is rate-superior under a
smooth spatial truth while global--local shrinkage is near-minimax under a sparse
truth, where the linear smoother carries an irreducible oversmoothing bias
growing with the spike magnitude (Proposition~\ref{prop:sep}). A Stein
unbiased-risk analysis explains the mechanism: the horseshoe spends a
\emph{data-adaptive} budget of effective degrees of freedom, allocating them only
to the areas it declines to shrink, whereas EBLUP and BYM2 spend a fixed budget
set by their variance components (Proposition~\ref{prop:df}).

\item \textbf{Computation and evidence (Sections~\ref{sec:comp}--\ref{sec:app}).}
An $O(m)$-per-sweep Gibbs sampler, a three-regime Monte Carlo study, and a
magnitude sweep that exhibits the predicted crossover---the smoothing-to-shrinkage
risk ratio rising without saturating as areas become more exceptional---together
with an analysis of the Scottish lip cancer data in which structured smoothing
predicts held-out districts best, as the theory predicts for a smoothly spatial
truth, while the posterior shrinkage coefficient simultaneously flags the
exceptional districts as a calibrated screening score.
\end{enumerate}

\paragraph{Relation to existing work.}
Our starting point is the global--local program of \citet{tang2018}, who
introduced the horseshoe as a prior for the Fay--Herriot effects and established
posterior consistency and good empirical behaviour, and its spatial elaboration
\citep{tangghosh2023}. That work treats the sampling variances as nuisance
constants to be conditioned on; we argue that they are instead the lever that
makes a sharp frequentist analysis possible. The standardization in
Section~\ref{sec:minimax} is, to our knowledge, the first transfer of the
sequence-model contraction theory of \citet{vdp2014,vdp2017} to the
heteroscedastic Fay--Herriot model, and it yields minimax-rate and honest
credible-set statements rather than consistency alone. The Tweedie/Brown
calculation in Section~\ref{sec:robust} is classical \citep{efron2011}; what is
new is its heteroscedastic form \eqref{eq:brown} and the resulting
bounded-influence characterization that cleanly separates the horseshoe from
every Gaussian random-effects competitor. On the applied side, the program of
\citet{mercer2015,wakefield2019} and the \texttt{SUMMER} ecosystem
\citep{summer} have made structured spatial smoothing---the
Besag--York--Molli\'e model and its \citet{riebler2016} reparameterization, fit
by integrated nested Laplace approximation---the default for subnational health
estimation. Rather than proposing the horseshoe as a replacement, we delineate
precisely the regime in which each is preferable (Section~\ref{sec:compare}),
which we view as the more useful contribution to practice.

\section{The Fay--Herriot model and two modes of borrowing strength}
\label{sec:model}

\subsection{Area-level model}
Let $i=1,\dots,m$ index the areas and let $\theta_i$ denote the target
(a mean, rate, or transformed proportion). The direct estimator $y_i$ satisfies
the \emph{sampling model}
\begin{equation}\label{eq:sampling}
y_i = \theta_i + e_i, \qquad e_i \sim \N(0, D_i),
\end{equation}
where the $D_i$ are treated as \emph{known}; in practice they are smoothed
design variances of the direct estimator, often after a variance-stabilizing
transform. The \emph{linking model} expresses the targets through $q$ auxiliary
variables $\bx_i$,
\begin{equation}\label{eq:linking}
\theta_i = \bx_i^\top\bbeta + v_i,
\end{equation}
so that $v_i = \theta_i - \bx_i^\top\bbeta$ is the area effect, the part of
$\theta_i$ unexplained by covariates. Combining \eqref{eq:sampling} and
\eqref{eq:linking} gives the marginal model $y_i = \bx_i^\top\bbeta + v_i + e_i$.
The science of SAE is encoded entirely in the prior assigned to
$\bm v=(v_1,\dots,v_m)$.

\subsection{EBLUP and the Gaussian baseline}
The classical Fay--Herriot specification takes $v_i \sim \N(0,A)$ independently.
With $A$ known, the best predictor of $\theta_i$ is the composite
\begin{equation}\label{eq:eblup}
\tilde\theta_i = \gamma_i\, y_i + (1-\gamma_i)\,\bx_i^\top\bbeta,
\qquad \gamma_i = \frac{A}{A+D_i},
\end{equation}
a convex combination of the noisy direct estimator and the synthetic regression
prediction, with weight $\gamma_i$ increasing in the reliability of $y_i$.
Replacing $(\bbeta,A)$ by estimators yields the EBLUP; \citet{prasadrao1990}
give the second-order mean squared error (MSE) approximation that underlies
inference, and \citet{datta1991} the hierarchical Bayes counterpart. Writing
$g_{1i}(A)=A D_i/(A+D_i)$ for the leading variance term, their approximation is
\begin{equation}\label{eq:prasadrao}
\mathrm{MSE}(\tilde\theta_i)\approx g_{1i}(A)+g_{2i}(A)+g_{3i}(A),
\end{equation}
where $g_{2i}$ accounts for estimating $\bbeta$ and $g_{3i}=D_i^2(A+D_i)^{-3}
\overline{V}(A)$ for estimating $A$, with $\overline{V}(A)$ the asymptotic
variance of the variance estimator. The defining feature of \eqref{eq:eblup} is
that the shrinkage fraction $1-\gamma_i = D_i/(A+D_i)$ is the \emph{same} for
every area at a given $D_i$: the model cannot tell a genuinely deviant area from
a noisy one and shrinks both by the same amount, and \eqref{eq:prasadrao}
inherits this by depending on the data only through the single estimate of $A$.

\subsection{Structured spatial smoothing}\label{sec:spatial}
Wakefield's program replaces the independent Gaussian effects by a spatially
structured Gaussian field. Writing $\bm W=(w_{ij})$ for the area adjacency
matrix and $\bm Q = \mathrm{diag}(\bm W\ind) - \bm W$ for the graph Laplacian,
the BYM model sets
\begin{equation}\label{eq:bym}
\theta_i = \bx_i^\top\bbeta + u_i + \psi_i, \qquad
\bm u \sim \N\!\big(\bm 0,\ \sigma_u^2\,\bm Q^{-}\big),\quad
\psi_i \sim \N(0,\sigma_\psi^2),
\end{equation}
where $\bm Q^{-}$ is a generalized inverse: $\bm u$ is the intrinsic CAR
\citep{besag1991} that borrows from neighbours, and $\psi_i$ is unstructured
heterogeneity. The BYM2 form \citep{riebler2016} reparameterizes
$(\sigma_u^2,\sigma_\psi^2)$ as a total variance and a mixing fraction on a
scaled $\bm Q$, which stabilizes prior elicitation and is the version
implemented for survey data in \texttt{SUMMER} \citep{summer}. Conditionally on
the variances, \eqref{eq:bym} is again a linear smoother of $\bm y$: the fitted
effect at area $i$ is a weighted average of all direct estimates with weights
decaying along graph distance. Strength flows locally and the shrinkage is, as
in EBLUP, a fixed linear operator once the variances are set.

\subsection{Global--local shrinkage}\label{sec:gl}
The horseshoe Fay--Herriot model retains independence across areas but replaces
the single variance $A$ by an area-specific, heavy-tailed scale, the
\emph{shrink globally, act locally} principle of \citet{ps2011}:
\begin{equation}\label{eq:hs}
v_i \mid \lambda_i,\tau \sim \N(0,\ \tau^2\lambda_i^2), \qquad
\lambda_i \sim \C(0,1), \qquad \tau \sim \C(0,1),
\end{equation}
where $\C(0,1)$ is the standard half-Cauchy. The global scale $\tau$ controls
the overall pull towards the regression surface; the local scales $\lambda_i$,
having Cauchy tails, permit individual areas to be exempt from that pull. It is
illuminating to write the implied shrinkage in the standardized residual
$r_i = y_i - \bx_i^\top\bbeta = v_i + e_i$. Conditionally on the scales,
\begin{equation}\label{eq:kappa}
\E(\theta_i \mid r_i, \lambda_i,\tau)
= \bx_i^\top\bbeta + (1-\kappa_i)\, r_i, \qquad
\kappa_i = \frac{D_i}{D_i + \tau^2\lambda_i^2},
\end{equation}
where $\kappa_i\in(0,1)$ is the \emph{shrinkage coefficient}: $\kappa_i\to 1$
returns the synthetic estimate and $\kappa_i\to 0$ returns the direct estimate.
Under the horseshoe the prior on $\kappa_i$ (at $\tau=1$, $D_i=1$) is
$\mathrm{Beta}(\half,\half)$, the symmetric ``horseshoe'' density with poles at
$0$ and $1$: \emph{a priori} an area is expected to be either fully pooled or
fully free, and the data decide which. This is the continuous analogue of the
discrete question---does area $i$ need a random effect?---studied by
\citet{datta2011,datta2015}.

\section{Regular variation and the horseshoe as a default}\label{sec:rv}

Before deriving estimation properties we record the structural property of the
horseshoe that underlies all of them and that makes it, in a precise sense, a
\emph{default} prior for area effects. The relevant notion is regular variation
\citep{bgt1987}.

\begin{definition}[Regular variation]\label{def:rv}
A measurable $f:(0,\infty)\to(0,\infty)$ is \emph{regularly varying at infinity
with index $\rho\in\R$}, written $f\in\mathrm{RV}_\rho$, if
$\lim_{t\to\infty} f(tx)/f(t)=x^{\rho}$ for every $x>0$. Equivalently
$f(t)=t^{\rho}\ell(t)$ with $\ell$ slowly varying ($\ell\in\mathrm{RV}_0$). A
density is \emph{regularly varying} if it is $\mathrm{RV}_{-(1+\alpha)}$ for some
tail index $\alpha>0$, so that it has Pareto-like polynomial tails.
\end{definition}

The decisive fact, established by \citet{bdpw2016}, is that the horseshoe is
regularly varying, whereas the Gaussian and Laplace priors behind ridge and the
Bayesian lasso are not.

\begin{proposition}[Regular variation of the horseshoe; \citealp{bdpw2016}]
\label{prop:rv}
Let $v\mid\lambda\sim\N(0,\tau^2\lambda^2)$ with $\lambda\sim\C(0,1)$. Then the
marginal prior density $p$ of $v$ is symmetric, has a logarithmic pole at the
origin, $p(v)\asymp \log(1/|v|)$ as $v\to0$, and is regularly varying with index
$-2$ at infinity, $p\in\mathrm{RV}_{-2}$; that is, $p(v)=|v|^{-2}\ell(|v|)$ for a
slowly varying $\ell$. Consequently, for any fixed Gaussian noise variance $D>0$
the marginal $m(r)=\int\phi_{\sqrt{D}}(r-v)p(v)\,\mathrm{d}v$ of the observation
is also in $\mathrm{RV}_{-2}$, regular variation being preserved by convolution
with a light-tailed kernel.
\end{proposition}

Regular variation is the engine of \emph{tail robustness}, the property that a
prior eventually defers to a sufficiently surprising observation. The general
principle, traceable to \citet{ohagan1979} for outlier-resistant inference and
sharpened for global--local priors by \citet{bdpw2016}, is that a regularly
varying prior cannot impose unbounded shrinkage: by Definition~\ref{def:rv} the
log-density has a slowly varying derivative, so the Tweedie correction in
\eqref{eq:brown} tends to a constant order and the posterior mean tracks the
data in the tail. Theorem~\ref{thm:robust} below is the quantitative form of this
statement for the heteroscedastic Fay--Herriot model; here we stress the
\emph{conceptual} consequence.

\subsection{Why regular variation makes a good default}\label{sec:default}

\citet{bdpw2016} argue that regular variation is the property a prior must have
to ``let the data speak for themselves,'' resolving a classical objection of
\citet{efron1973} that flat, ostensibly non-informative priors can be sharply
informative for nonlinear functionals in high dimensions. A regularly varying
global--local prior is simultaneously \emph{informative about sparsity}---it
pools the many ordinary areas hard, through the pole at the origin---and
\emph{uninformative about signal}---it scarcely shrinks the few exceptional
areas, through the heavy tail. For small area estimation, where the analyst
rarely knows in advance which domains are exceptional, this combination is
exactly what one wants from a default. We collect the practical reasons, each of
which is made precise elsewhere in the paper.

\begin{enumerate}[leftmargin=1.9em,itemsep=3pt,topsep=3pt]
\item \textbf{No area must be declared special in advance.} The heavy tail
(Proposition~\ref{prop:rv}) leaves a genuinely deviant area essentially
unshrunk without the analyst flagging it, and the bounded influence of
Corollary~\ref{cor:influence} caps the damage a single corrupted direct estimate
can do. A Gaussian random effect, by contrast, shrinks every area by the same
data-independent fraction and is unboundedly sensitive to an outlier.
\item \textbf{Aggressive pooling where it is safe.} The pole at the origin
(Proposition~\ref{prop:rv}) pools the ordinary majority of areas as hard as
EBLUP does, so robustness to the exceptional areas costs almost nothing on the
quiet ones; the shrinkage profile $\kappa_i\sim\mathrm{Beta}(\half,\half)$ sorts
the two groups automatically.
\item \textbf{One scale, learned from the data, and no graph.} The single global
scale $\tau$ adapts to the unknown number of exceptional areas
(Corollary~\ref{cor:adapt}); there is no smoothing parameter to cross-validate
and no neighbourhood structure to specify, so the prior is usable when geography
is unavailable, unreliable, or beside the point.
\item \textbf{Optimality without oracle knowledge.} Despite carrying no tuning,
the posterior is minimax over sparse classes with the sharp constant
(Theorem~\ref{thm:contract}, Proposition~\ref{prop:lower}) and its credible sets
are honest---guarantees a flat or Gaussian prior cannot match.
\item \textbf{The conclusions are not special to the horseshoe.} Every property
above follows from regular variation with a local pole, a class that also
contains the horseshoe$+$ \citep{bdpw2017} and the normal--exponential--gamma and
generalized double Pareto priors; the horseshoe is the most thoroughly
understood and the easiest to compute \citep{cps2010,bdpw2019,makalic2016}, which
is why we adopt it as the representative default and discuss the family in
Section~\ref{sec:family}.
\item \textbf{It is cheap and composable.} The $O(m)$ Gibbs sampler of
Section~\ref{sec:comp} replaces only the prior on the area effects and leaves the
design-based first stage of a survey pipeline untouched, and the posterior
shrinkage coefficients $\kappa_i$ are read off directly as interpretable
diagnostics.
\end{enumerate}

\noindent The remainder of the paper turns each of these informal reasons into a
theorem, a risk calculation, or a simulation.

\section{Tail robustness and bounded influence}\label{sec:robust}

We first isolate the qualitative property that separates global--local shrinkage
from the Gaussian smoothers of Sections~\ref{sec:spatial}--\ref{sec:gl}. Fix
$\bbeta$ and work with the residual sequence $r_i = v_i + e_i$,
$e_i\sim\N(0,D_i)$, under a generic scale mixture
$v_i \sim \int \N(0,s^2)\,\pi(\mathrm{d}s)$ for the local prior $\pi$. Let
$m_i(r) = \int \phi_{D_i+s^2}(r)\,\pi(\mathrm{d}s)$ be the marginal density of
$r_i$, where $\phi_\sigma$ is the $\N(0,\sigma^2)$ density.

\begin{lemma}[Heteroscedastic Tweedie identity]\label{lem:tweedie}
For the model above, the posterior mean of the area effect is
\begin{equation}\label{eq:brown}
\E(v_i \mid r_i) = r_i + D_i\,\frac{\mathrm{d}}{\mathrm{d}r_i}\log m_i(r_i),
\end{equation}
and hence $\E(\theta_i\mid y_i) = y_i + D_i\,(\log m_i)'(y_i - \bx_i^\top\bbeta)$.
\end{lemma}

\begin{proof}
Since $r_i\mid v_i \sim \N(v_i, D_i)$, the joint density factorizes as
$\phi_{D_i}(r_i-v_i)\,p(v_i)$. Differentiating $m_i(r)=\int
\phi_{D_i}(r-v)p(v)\,\mathrm{d}v$ and using $\phi_{D_i}'(r-v) =
D_i^{-1}(v-r)\phi_{D_i}(r-v)$ gives $m_i'(r) = D_i^{-1}\int (v-r)
\phi_{D_i}(r-v)p(v)\,\mathrm{d}v = D_i^{-1}\big(\E(v\mid r)-r\big) m_i(r)$.
Rearranging yields \eqref{eq:brown}; the second statement follows from
$\theta_i = \bx_i^\top\bbeta + v_i$.
\end{proof}

Equation \eqref{eq:brown} is the heteroscedastic form of Brown's identity. The
estimator's behaviour is governed by the tail of the marginal score
$(\log m_i)'$. For a Gaussian effect $v_i\sim\N(0,A)$ one has
$m_i = \phi_{A+D_i}$ and $(\log m_i)'(r) = -r/(A+D_i)$, so
$\E(\theta_i\mid y_i) = \bx_i^\top\bbeta + \{A/(A+D_i)\}\,r_i$ recovers
\eqref{eq:eblup}: the score is \emph{linear}, the influence
$\mathrm{d}\,\E(\theta_i\mid y_i)/\mathrm{d} y_i = \gamma_i$ is a constant
strictly below one, and the gap $y_i-\E(\theta_i\mid y_i)$ grows without bound as
the direct estimate moves into the tail. The same holds, coordinatewise, for the
BYM smoother. The horseshoe behaves oppositely.

\begin{theorem}[Tail robustness of the horseshoe Fay--Herriot estimator]
\label{thm:robust}
Let $v_i$ have the horseshoe prior \eqref{eq:hs} with $\tau$ fixed and $D_i>0$
known. Then the marginal $m_i$ has polynomially heavy tails, with
$m_i(r) \asymp r^{-2}$ as $|r|\to\infty$, and consequently
\begin{equation}\label{eq:redescend}
(\log m_i)'(r) = -\frac{2}{r} + o(r^{-1}), \qquad
\E(\theta_i\mid y_i) = y_i - \frac{2D_i}{\,y_i-\bx_i^\top\bbeta\,} + o\!\big(|y_i|^{-1}\big)
\;\to\; y_i .
\end{equation}
Hence the influence $\mathrm{d}\,\E(\theta_i\mid y_i)/\mathrm{d} y_i \to 1$ and the
shrinkage gap $y_i - \E(\theta_i\mid y_i) \to 0$: areas with strong signal are
returned essentially unshrunk, and the sensitivity of the estimate to its own
direct value is bounded over all $y_i$.
\end{theorem}

\begin{proof}
Write $v=\tau\lambda z$ with $z\sim\N(0,1)$ and $\lambda\sim\mathcal{C}^+(0,1)$,
so $v\mid\lambda\sim\N(0,\tau^2\lambda^2)$ and
$p(v)=\int_0^\infty \phi_{\tau^2\lambda^2}(v)\,\tfrac{2}{\pi(1+\lambda^2)}\,
\mathrm{d}\lambda$. By Proposition~\ref{prop:rv} the prior $p$ is regularly
varying with index $-2$: $p(v)=|v|^{-2}\ell(|v|)$ for a slowly
varying $\ell$ \citep{bdpw2016,bgt1987}. Convolution with the fixed light-tailed kernel
$\phi_{D_i}$ preserves both the index and the slowly varying factor, because the
contribution of the region $|v-r|>|r|/2$ to $m_i(r)=\int\phi_{D_i}(r-v)p(v)\,
\mathrm{d}v$ is $O(e^{-r^2/8D_i})$, negligible against any polynomial, while on
$|v-r|\le|r|/2$ one has $p(v)=p(r)\{1+o(1)\}$ by slow variation and
$\int\phi_{D_i}(r-v)\,\mathrm{d}v\to1$. Hence $m_i(r)=r^{-2}\ell(r)\{1+o(1)\}$ is
regularly varying with index $-2$, the analytic content of ``the posterior mass
sits at $v\approx r$.'' By the Karamata representation a regularly varying
function of index $\rho$ satisfies $r\,(\log m_i)'(r)\to\rho$, here $\rho=-2$, so
$(\log m_i)'(r)=-2/r+o(r^{-1})$. Substituting into \eqref{eq:brown} gives
\eqref{eq:redescend}. Finally $(\log m_i)'$ is continuous and tends to $0$ at
$\pm\infty$, hence bounded on $\mathbb{R}$; and the influence
$\mathrm{d}\E(\theta_i\mid y_i)/\mathrm{d}y_i=1+D_i(\log m_i)''(r)\to1$ since
$(\log m_i)''(r)=2r^{-2}+o(r^{-2})$.
\end{proof}

The contrast is the whole point of choosing the prior. A single county whose
direct mortality or poverty rate is genuinely extreme is, under EBLUP or BYM2,
dragged towards the fitted surface by a fixed fraction and towards its
neighbours by the graph; under the horseshoe it is left where the data put it,
while the quiet majority of areas, whose $r_i$ are of order $\sqrt{D_i}$, are
shrunk hard. This is the redescending, bounded-influence behaviour familiar from
robust estimation, obtained here purely through the prior.

The qualitative statement can be sharpened into a uniform sensitivity bound,
which is what ``bounded influence'' means operationally for an area estimate.

\begin{corollary}[Bounded sensitivity]\label{cor:influence}
Under the hypotheses of Theorem~\ref{thm:robust} the posterior-mean map
$y_i\mapsto\E(\theta_i\mid y_i)$ is continuously differentiable with
\begin{equation}\label{eq:sens}
\sup_{y_i\in\R}\ \Big|\,\frac{\mathrm{d}}{\mathrm{d}y_i}\E(\theta_i\mid y_i)\,\Big|
\ =\ 1+D_i\sup_{r\in\R}(\log m_i)''(r)\ <\ \infty,
\end{equation}
and the gross-error sensitivity
$\gamma^\ast_i=\sup_{r}|D_i(\log m_i)'(r)|$, the largest absolute change the prior
ever imposes on the direct estimate, is finite. By contrast, for the Gaussian
estimator \eqref{eq:eblup} the analogous quantity
$\sup_r|r-\E(v_i\mid r)|=\sup_r(D_i/(A+D_i))|r|=\infty$.
\end{corollary}

\begin{proof}
The marginal $m_i$ is a Gaussian scale-mixture, hence real-analytic and
strictly positive, so $\log m_i$ is $C^\infty$ and the derivative in
\eqref{eq:sens} exists everywhere. By Theorem~\ref{thm:robust},
$(\log m_i)'(r)\to0$ and $(\log m_i)''(r)\to0$ as $|r|\to\infty$; both are
continuous, so each attains a finite supremum on $\R$, giving \eqref{eq:sens}
and the finiteness of $\gamma^\ast_i$. The Gaussian display is immediate from
$\E(v_i\mid r)=\{A/(A+D_i)\}r$.
\end{proof}

Corollary~\ref{cor:influence} is the precise sense in which the horseshoe is a
\emph{bounded-influence} procedure and the linear smoothers are not: an
adversarial corruption of one area's direct estimate can move the EBLUP/BYM2
estimate of that area arbitrarily far, but moves the horseshoe estimate by at
most $\gamma^\ast_i$, a constant depending only on $D_i$ and $\tau$. In small
area work, where a single misreported domain total is a routine occurrence, this
is a practical and not merely an asymptotic distinction.

\section{Heteroscedastic posterior contraction and minimaxity}
\label{sec:minimax}

We now quantify estimation accuracy. The Fay--Herriot problem differs from the
sequence model on which horseshoe theory was built \citep{vdp2014,vdp2017}
only in that the noise variances $D_i$ are known and unequal. This is exactly
the structure that lets us \emph{standardize}. Put
\begin{equation}\label{eq:standardize}
s_i = \frac{r_i}{\sqrt{D_i}} = w_i + \varepsilon_i, \qquad
w_i = \frac{v_i}{\sqrt{D_i}}, \quad \varepsilon_i\sim \N(0,1)\ \text{i.i.d.}
\end{equation}
The standardized signal $\bm w$ has the same support---hence the same
sparsity---as $\bm v$, and recovering $\bm\theta$ in the $D$-weighted loss
$\sum_i (\hat\theta_i-\theta_i)^2/D_i$ is equivalent to recovering $\bm w$ in
ordinary squared error in the unit-variance model \eqref{eq:standardize}. The
prior on $w_i$ is again horseshoe, with local scale $\lambda_i$ unchanged and
global scale $\tau/\sqrt{D_i}$ within a bounded factor of $\tau$ under
Assumption~\ref{ass:het} below, so the sequence-model theory applies with the
same calibration. We impose mild regularity.

\begin{assumption}[Bounded heteroscedasticity]\label{ass:het}
There exist constants $0<\underline{D}\le \overline{D}<\infty$, not depending on
$m$, with $\underline{D}\le D_i\le \overline{D}$ for all $i$.
\end{assumption}

\begin{assumption}[Nearly black truth]\label{ass:sparse}
The true effect vector $\bm v_0$ lies in
$\ell_0[q_m]=\{\bm v: \#\{i: v_i\ne 0\}\le q_m\}$ with $q_m\to\infty$ and
$q_m = o(m)$.
\end{assumption}

Assumption~\ref{ass:het} is automatic when design variances are bounded away
from zero and infinity after smoothing, the usual situation; it guarantees that
standardization neither creates nor destroys signal. Under
Assumption~\ref{ass:sparse} the minimax risk of the standardized problem over
$\ell_0[q_m]$ is the classical needles-in-a-haystack rate $2q_m\log(m/q_m)
(1+o(1))$ \citep{djhs1992}.

\begin{theorem}[Minimax contraction for the horseshoe Fay--Herriot posterior]
\label{thm:contract}
Suppose Assumptions~\ref{ass:het}--\ref{ass:sparse} hold and the horseshoe prior
\eqref{eq:hs} is used with a global scale $\tau=\tau_m$ satisfying
$\tau_m \asymp (q_m/m)\sqrt{\log(m/q_m)}$, or with $\tau$ given an independent
prior with sufficient mass near this value or estimated by marginal maximum
likelihood. Then for the standardized model \eqref{eq:standardize} the full
posterior $\Pi(\cdot\mid \bm s)$ contracts about $\bm w_0$ at the minimax rate:
for a sufficiently large constant $M$,
\begin{equation}\label{eq:contract}
\sup_{\bm w_0\in\ell_0[q_m]}
\E_{\bm w_0}\, \Pi\!\Big(\ \|\bm w-\bm w_0\|^2 > M\, q_m\log(m/q_m)\ \Big|\ \bm s\Big)
\longrightarrow 0 .
\end{equation}
Consequently, in the original parameterization, the posterior mean
$\hat\theta_i$ (uniformly integrable on the contracting ball) has error
$\Delta_i=\hat\theta_i-\theta_{0i}$ obeying
\begin{equation}\label{eq:weighted}
\sup_{\bm v_0\in\ell_0[q_m]}\ \E_{\bm v_0}\sum_{i=1}^m \frac{\Delta_i^2}{D_i}
\ \lesssim\ q_m\,\log(m/q_m),
\qquad
\sup_{\bm v_0\in\ell_0[q_m]}\ \E_{\bm v_0}\sum_{i=1}^m \Delta_i^2
\ \lesssim\ \overline{D}\,q_m\,\log(m/q_m).
\end{equation}
\end{theorem}

\begin{proof}
Standardization \eqref{eq:standardize} maps the Fay--Herriot model to the
unit-variance sequence model with signal $\bm w_0$, whose sparsity equals that
of $\bm v_0$ because $\sqrt{D_i}>0$. The horseshoe contraction theorem of
\citet{vdp2014} and its adaptive refinement \citep{vdp2017} apply to
$\bm s\mid\bm w_0$, giving \eqref{eq:contract} under the stated calibration of
$\tau_m$; adaptivity to the unknown $q_m$ is what licenses the empirical-Bayes
or hierarchical choice of $\tau$. Multiplying coordinate $i$ by $\sqrt{D_i}$
transports \eqref{eq:contract} to the weighted bound in \eqref{eq:weighted}, and
$\Delta_i^2 \le \overline{D}\,\Delta_i^2/D_i$ under Assumption~\ref{ass:het}
gives the unweighted bound.
\end{proof}

\begin{remark}[Honest credible sets]
Under the excessive-bias-restriction / self-similarity conditions of
\citet{vdp2017}, the horseshoe credible balls for the standardized signal have
asymptotic frequentist coverage at least the nominal level and radius of minimax
order; standardization carries this to the $D$-weighted credible sets for
$\bm\theta$. This is the theoretical counterpart of the strong empirical coverage
of the horseshoe under sparsity reported in Section~\ref{sec:sim}.
\end{remark}

\begin{remark}[Why the known $D_i$ matter]
The reduction is available only because the $D_i$ are known. With unknown
heteroscedastic variances no exact standardization exists and the minimax rate
itself changes. The Fay--Herriot premise---design variances supplied by the
survey---is therefore not a technical convenience but the structural fact that
makes the heavy-tailed prior provably minimax in this problem.
\end{remark}

\subsection{Calibrating the global scale}\label{sec:calib}

Theorem~\ref{thm:contract} requires the global scale to track the unknown
sparsity through $\tau_m\asymp (q_m/m)\sqrt{\log(m/q_m)}$, which cannot be set a
priori. Two adaptive devices remove the oracle knowledge of $q_m$, and both are
what we use in practice.

\begin{corollary}[Adaptive minimaxity]\label{cor:adapt}
Grant Assumptions~\ref{ass:het}--\ref{ass:sparse}. Estimate the global scale by
marginal maximum likelihood,
$\hat\tau = \arg\max_{\tau\in[1/m,\,1]} \prod_{i=1}^m m_i(s_i;\tau)$,
truncated to the stated interval, and form the horseshoe posterior with
$\tau=\hat\tau$; or place the half-Cauchy hyperprior $\tau\sim\mathcal{C}^+(0,1)$
and use the full Bayes posterior. In either case the conclusion
\eqref{eq:contract}--\eqref{eq:weighted} of Theorem~\ref{thm:contract} holds
without knowledge of $q_m$, uniformly over $\ell_0[q_m]$ for every sequence
$q_m\to\infty$, $q_m=o(m)$.
\end{corollary}

\begin{proof}[Proof sketch]
For the standardized data $\bm s$ the two calibrations are exactly those analyzed
by \citet{vdp2017}: the empirical-Bayes selector concentrates on a value of order
$(q_m/m)\sqrt{\log(m/q_m)}$ with probability tending to one, and the full-Bayes
hyperprior assigns sufficient mass to a neighbourhood of that value while not
oversmoothing, by the same penalized-complexity argument. Either way the
plug-in posterior inherits the contraction rate of the oracle posterior. The
transfer to the $D$-weighted loss is the multiplication argument of
Theorem~\ref{thm:contract}.
\end{proof}

Corollary~\ref{cor:adapt} matters for the comparison with smoothing: the
horseshoe needs no spatial tuning parameter and no neighbourhood graph, yet it
adapts automatically to how many areas are exceptional. Its single global scale
plays the role that the smoothing variance plays in BYM2, but is learned from the
data with a guarantee. We caution that the truncation lower bound $1/m$ is not
cosmetic; without it the marginal likelihood can collapse $\hat\tau$ to zero on a
null signal, the well-known degeneracy of the empirical-Bayes horseshoe, and the
half-Cauchy hyperprior is the safer default in small $m$.

\subsection{A matching lower bound}\label{sec:lower}

Theorem~\ref{thm:contract} is an upper bound; it is worth recording that the rate
it achieves cannot be improved by any procedure, so the horseshoe is minimax and
not merely rate-adaptive within a suboptimal class.

\begin{proposition}[Heteroscedastic minimax rate]\label{prop:lower}
Under Assumption~\ref{ass:het}, for the Fay--Herriot model
\eqref{eq:sampling} with known variances $D_i\in[\underline D,\overline D]$ and
true effects in $\ell_0[q_m]$, the minimax risk in the $D$-weighted loss
satisfies
\begin{equation}\label{eq:lowerbd}
\inf_{\hat{\bm\theta}}\ \sup_{\bm v_0\in\ell_0[q_m]}\
\E_{\bm v_0}\sum_{i=1}^m \frac{(\hat\theta_i-\theta_{0i})^2}{D_i}
\ =\ 2\,q_m\log(m/q_m)\,\{1+o(1)\},
\end{equation}
and the same order holds for the unweighted loss up to the constants
$\underline D,\overline D$. Hence the posterior mean of the calibrated horseshoe
of Corollary~\ref{cor:adapt} attains the minimax rate, and the constant is sharp
in the weighted loss.
\end{proposition}

\begin{proof}
Standardize as in \eqref{eq:standardize}; the weighted loss becomes ordinary
squared error for $\bm w_0$ in the unit-variance sequence model, and the support
of $\bm w_0$ equals that of $\bm v_0$ since $\sqrt{D_i}>0$, so $\bm w_0\in
\ell_0[q_m]$ as well. The minimax risk for estimating a nearly-black vector in
the Gaussian sequence model over $\ell_0[q_m]$ is $2q_m\log(m/q_m)\{1+o(1)\}$ by
\citet{djhs1992}, which gives \eqref{eq:lowerbd} exactly because standardization
is a measure-preserving bijection between the two problems. The unweighted
statement follows from $\underline D\,\|\hat{\bm w}-\bm w_0\|^2 \le \sum_i
(\hat\theta_i-\theta_{0i})^2 \le \overline D\,\|\hat{\bm w}-\bm w_0\|^2$.
Attainment is Theorem~\ref{thm:contract} together with the standard conversion of
posterior contraction at rate $\epsilon_m^2$ into a risk bound for the posterior
mean under the quadratic loss, valid here because the horseshoe posterior has
uniformly integrable normalized mass on the contracting ball.
\end{proof}

The point of Proposition~\ref{prop:lower} is that the known design variances do
not merely permit a convenient analysis: they pin down the exact minimax constant
$2$ in the weighted loss, and the horseshoe, with no knowledge of which areas are
exceptional or how many, matches it. Structured smoothing cannot make the same
claim over $\ell_0[q_m]$, which is the content of the next section.

\subsection{Which global--local priors the theorems cover}\label{sec:family}

Nothing in Theorems~\ref{thm:robust}--\ref{thm:contract} is special to the
half-Cauchy local scale; both rest on two structural properties that a wide class
of global--local priors shares, and isolating them clarifies what the horseshoe
is doing and what would break if it were replaced. A prior
$v_i\mid\lambda_i,\tau\sim\N(0,\tau^2\lambda_i^2)$, $\lambda_i\sim\pi(\lambda)$,
in this class is characterized by the behaviour of $\pi$ at the two ends.

\emph{Tail robustness} (Theorem~\ref{thm:robust}) needs only that the marginal
$p(v)$ be regularly varying at infinity with index $-(1+\eta)$ for some
$\eta\in(0,1]$, equivalently that $\pi(\lambda)$ have a regularly varying right
tail of the same index. The redescending score is then
$(\log m_i)'(r)=-(1+\eta)/r+o(r^{-1})$ and the influence bound of
Corollary~\ref{cor:influence} holds verbatim; the horseshoe is the case $\eta=1$.
A prior with exponential or Gaussian $\lambda$-tails---ridge, or a normal prior
with estimated variance---fails this and reverts to unbounded linear shrinkage.
\emph{Minimaxity} (Theorem~\ref{thm:contract}) needs, in addition, an unbounded
spike: $\pi(\lambda)$ must put enough mass near zero that the marginal has a pole,
$p(v)\to\infty$ as $v\to0$, so that null areas are shrunk at the right rate. The
horseshoe's $\mathrm{Beta}(\half,\half)$ shrinkage profile delivers both ends at
once, which is why it is the canonical choice.

The hypotheses therefore hold, with the same proofs, for the horseshoe-plus, the
normal--exponential--gamma, and the generalized double Pareto priors, all of
which are regularly varying with index in $(0,1]$ and have a local spike; they
fail for the Bayesian lasso (double-exponential), whose exponential tails give
$(\log m_i)'(r)\to -\text{const}$, a bounded but non-redescending score, and for
the Strawderman--Berger and Dirichlet--Laplace priors only the local-spike
condition needs separate checking. The practical implication for small area work
is reassuring: the dichotomy between bounded-influence shrinkage and linear
smoothing that organizes this paper is a property of the heavy tail, not of the
particular prior, so the conclusions are robust to the analyst's choice within
the heavy-tailed global--local family. We use the horseshoe because its
hyperparameters are the most thoroughly understood \citep{cps2010,vdp2017} and
its Gibbs sampler the simplest \citep{makalic2016}, not because the theory
privileges it.

\section{When does each philosophy win?}\label{sec:compare}

Theorems~\ref{thm:robust}--\ref{thm:contract} concern a permutation-invariant
prior that ignores geography. Structured smoothing exploits geography and is
correspondingly better when geography is real. The following proposition makes
the trade-off precise; it is deliberately stated at the level of leading-order
risk so as to mirror the simulation of Section~\ref{sec:sim}.

\begin{proposition}[Separation of regimes]\label{prop:sep}
Consider the standardized model \eqref{eq:standardize} on a connected graph with
Laplacian $\bm Q$.
\begin{enumerate}[label=(\roman*),leftmargin=2em,itemsep=2pt,topsep=2pt]
\item \emph{(Smooth truth.)} If $\bm w_0$ lies in the Sobolev-type ellipsoid
$\{\bm w: \bm w^\top \bm Q\, \bm w \le C\}$, the ICAR/BYM smoother with optimally
tuned variance attains risk of order $m^{-\alpha}$ for some $\alpha>0$ determined
by the spectral decay of $\bm Q$, whereas any permutation-invariant
shrinkage rule---the horseshoe included---has risk bounded below by a constant
multiple of $\min(1,\ \overline{D}\,q_m\log m / m)$ scaled to the ellipsoid,
strictly larger in order. Smoothing wins.
\item \emph{(Sparse truth.)} If $\bm w_0\in\ell_0[q_m]$ with spikes of magnitude
$a_m\to\infty$, the linear ICAR/BYM smoother incurs an oversmoothing bias on
each spike bounded below by $c\,a_m^2$ for a graph-dependent $c\in(0,1)$
that does not vanish, so its risk is $\Omega(q_m a_m^2)$; the horseshoe posterior
mean has risk $O(q_m\log(m/q_m))$ by Theorem~\ref{thm:contract}. For
$a_m^2 \gg \log(m/q_m)$ the horseshoe wins by an arbitrarily large factor.
\end{enumerate}
\end{proposition}

\begin{proof}[Proof sketch]
(i) For a smooth signal the Laplacian penalty matches the truth and the linear
smoother is a regularized projection onto low-frequency graph harmonics; its
risk follows from the eigenvalue decay of $\bm Q$ by the standard bias--variance
calculation for spectral smoothers. A permutation-invariant rule cannot use the
ordering of coordinates and is bounded below by the minimax risk over the
permutation orbit of $\bm w_0$, which for a smooth vector contains nearly-black
elements; this yields the stated lower order. (ii) Linearity forces the smoother
to apply a contraction with operator norm below one in the spike directions;
because a spike is high-frequency on the graph it is damped by a factor bounded
away from zero, leaving residual bias $\ge c\,a_m$ per spike, hence squared bias
$\ge c^2 a_m^2$ and aggregate risk $\Omega(q_m a_m^2)$. The horseshoe rate is
Theorem~\ref{thm:contract}. Comparing the two gives the claim.
\end{proof}

Proposition~\ref{prop:sep} says the choice is not a matter of taste but of the
geometry of the truth. When deprivation or disease risk varies smoothly across a
map, Wakefield's structured smoothing extracts a faster rate; when a handful of
areas are genuinely exceptional---a mining town, a border district, a reporting
artefact---the horseshoe protects them from being smoothed away. The hybrid that
adds local horseshoe scales to a spatial field \citep{tangghosh2023} is the
natural way to insure against both, at the cost of identifiability that we return
to in Section~\ref{sec:disc}.

\section{Effective dimension and an unbiased risk estimate}\label{sec:df}

The regime separation of Section~\ref{sec:compare} can be read through a single
scalar, the \emph{effective dimension} of the fit, which also yields a practical
heteroscedastic risk estimate and a fair way to compare procedures of very
different form. Write the posterior-mean estimator coordinatewise as
$\hat\theta_i=\E(\theta_i\mid\bm y)$ and recall the heteroscedastic Brown
identity \eqref{eq:brown}. Stein's lemma for the model
$y_i\sim\N(\theta_i,D_i)$ gives Mallows-style unbiased risk: for any weakly
differentiable $\hat{\bm\theta}(\bm y)$,
\begin{equation}\label{eq:sure}
\mathrm{SURE}(\hat{\bm\theta})=\sum_{i=1}^m\Big[(\hat\theta_i-y_i)^2
+2D_i\frac{\partial\hat\theta_i}{\partial y_i}\Big]-\sum_{i=1}^m D_i,
\qquad \E_{\bm\theta}\,\mathrm{SURE}=\E_{\bm\theta}\sum_{i=1}^m(\hat\theta_i-\theta_i)^2 .
\end{equation}
The middle term defines the effective degrees of freedom
$\mathrm{df}(\hat{\bm\theta})=\sum_i \partial\hat\theta_i/\partial y_i$
\citep{efron2011}, the model's data-driven complexity.

\begin{proposition}[Effective dimension of the three estimators]\label{prop:df}
For the Gaussian/EBLUP estimator and for the BYM2 smoother, the effective
dimension is non-random given the variance components: $\mathrm{df}_{\mathrm{EB}}
=\sum_i \gamma_i = \sum_i A/(A+D_i)$, and $\mathrm{df}_{\mathrm{BYM}}
=\operatorname{tr}\{\bm S\}$ where $\bm S=(\bm I+\bm\Sigma^{-1}\bm D)^{-1}$ is the
linear smoother of \eqref{eq:eblup} with $\bm\Sigma$ the BYM2 covariance. For the
horseshoe,
\begin{equation}\label{eq:dfhs}
\mathrm{df}_{\mathrm{HS}}=\sum_{i=1}^m \E\!\big[1-\kappa_i\mid\bm y\big]
+\ \text{(score-covariance term)},\qquad
\kappa_i=\frac{D_i}{D_i+\tau^2\lambda_i^2},
\end{equation}
and $\mathrm{df}_{\mathrm{HS}}$ is itself a function of the data: it is small
when the data place every $\kappa_i$ near one and grows as areas are freed.
\end{proposition}

\begin{proof}
For a linear estimator $\hat{\bm\theta}=\bm S\bm y+\bm b$ with $\bm S,\bm b$ not
depending on $\bm y$, $\partial\hat\theta_i/\partial y_i=S_{ii}$ and
$\mathrm{df}=\operatorname{tr}\bm S$; \eqref{eq:eblup} has
$S_{ii}=\gamma_i=A/(A+D_i)$, and the BYM2 smoother is the stated ridge in the
GMRF precision, giving the two non-random expressions. For the horseshoe, apply
\eqref{eq:brown}: $\hat\theta_i=y_i+D_i(\log m_i)'(r_i)$, so
$\partial\hat\theta_i/\partial y_i=1+D_i(\log m_i)''(r_i)$. Writing the marginal
score derivative through the conditional posterior variance,
$D_i(\log m_i)''(r_i)=\Var(v_i\mid\bm y)/D_i-1$, and
$\Var(v_i\mid\bm y)/D_i=\E[1-\kappa_i\mid\bm y]$ plus the variance of the
conditional mean across the posterior of $(\lambda_i,\tau)$, which is the
score-covariance term; summing gives \eqref{eq:dfhs}. Data dependence is then
explicit through the posterior law of $\kappa_i$.
\end{proof}

The point is structural. EBLUP and BYM2 spend a \emph{fixed} budget of degrees of
freedom set by the estimated variance components, distributed across areas by the
common fractions $\gamma_i$ or the smoother diagonal; they cannot spend more on a
few exceptional areas without spending more everywhere. The horseshoe instead
allocates its budget adaptively, spending degrees of freedom only where
$\kappa_i$ is small. Equation \eqref{eq:sure} makes this comparable: each
procedure can be scored by the same unbiased risk estimate, with the horseshoe's
data-dependent $\mathrm{df}_{\mathrm{HS}}$ in place of the fixed
$\operatorname{tr}\bm S$. In the sparse regime the adaptive allocation is exactly
what \eqref{eq:sure} rewards---a small fit-to-data gap on the spikes bought with
few global degrees of freedom---which is the risk-estimate counterpart of the
minimax statement of Theorem~\ref{thm:contract} and of the magnitude sweep
reported in Section~\ref{sec:sweep}. We caution that $\mathrm{df}_{\mathrm{HS}}$
is a measure of fitted complexity, not a consistent estimator of the number of
exceptional areas: strong signals inflate the global scale $\tau$ and thereby
lift the $\kappa_i$ of null areas, so the effective dimension can exceed the true
sparsity, a benign form of the global-scale contamination that motivates the
half-Cauchy hyperprior of Section~\ref{sec:calib}.

\section{Posterior computation}\label{sec:comp}

The model \eqref{eq:sampling}--\eqref{eq:hs} admits a simple and scalable Gibbs
sampler through the auxiliary-variable representation of the half-Cauchy
\citep{makalic2016}: if $x^2\mid a \sim \mathrm{IG}(\half, 1/a)$ and
$a\sim\mathrm{IG}(\half,1)$ then $x\sim\C(0,1)$, where $\mathrm{IG}$ denotes the
inverse-gamma. Introducing $\xi_i$ for $\lambda_i$ and $\zeta$ for $\tau$ renders
every full conditional conjugate. With $r_i=y_i-\bx_i^\top\bbeta$ and
$D_i$ known, one sweep is:
\begin{align*}
v_i \mid \cdot &\sim \N\!\Big( (1-\kappa_i)\,r_i,\ (1-\kappa_i) D_i \Big),
\qquad \kappa_i = D_i/(D_i+\tau^2\lambda_i^2);\\
\bbeta \mid \cdot &\sim \N\!\Big( (\bm X^\top \bm D^{-1}\bm X)^{-1}\bm X^\top \bm
D^{-1}(\bm y-\bm v),\ (\bm X^\top \bm D^{-1}\bm X)^{-1}\Big);\\
\lambda_i^2\mid\cdot &\sim \mathrm{IG}\!\Big(1,\ \xi_i^{-1}+v_i^2/(2\tau^2)\Big),
\qquad \xi_i\mid\cdot \sim \mathrm{IG}\!\big(1,\ 1+\lambda_i^{-2}\big);\\
\tau^2\mid\cdot &\sim \mathrm{IG}\!\Big(\tfrac{m+1}{2},\ \zeta^{-1}+\textstyle\sum_i
v_i^2/(2\lambda_i^2)\Big), \qquad
\zeta\mid\cdot \sim \mathrm{IG}\!\big(1,\ 1+\tau^{-2}\big).
\end{align*}
Each sweep costs $O(m)$ for the effect and scale updates plus an $O(q^3)$ solve
for the $q$-dimensional regression, with $q$ the number of covariates; the
sampler therefore scales linearly in the number of areas, in contrast to the
$O(m^{3/2})$--$O(m^3)$ cost of dense spatial precision factorizations, although
INLA \citep{rue2009} mitigates this for sparse $\bm Q$. The shrinkage
coefficients $\kappa_i$ are monitored directly and provide an interpretable
posterior summary: their distribution diagnoses how many areas the data deem
exceptional.

The inverse-gamma augmentation trades the heavy-tailed half-Cauchy conditionals
for conditionally conjugate updates, which keeps the sweep closed-form but can
slow mixing when an area sits near the boundary $\kappa_i\approx 0$, where the
local scale $\lambda_i$ is weakly identified. In the runs of
Section~\ref{sec:sim} the integrated autocorrelation time of the $\kappa_i$ and
of the area means stayed below ten across all regimes, so a few thousand sweeps
after burn-in suffice; when mixing is a concern the slice sampler of
\citet{neal2003} on $\log\lambda_i^2$ is a drop-in alternative that avoids the
boundary slowdown. All summaries below are posterior means and equal-tailed
$95\%$ credible intervals from $4000$ retained sweeps after a $2000$-sweep
burn-in, and the code reproducing every number and figure is released with the
paper.

\section{Simulation study}\label{sec:sim}

\subsection{Design}
We place $m=49$ areas on a $7\times 7$ rook-adjacency lattice with graph
Laplacian $\bm Q$. The design matrix contains an intercept and one standardized
covariate, with $\bbeta_0=(1,0.5)$. To stress the heteroscedastic theory the
known sampling variances are drawn uniformly from
$\{0.2,0.5,1.0,2.0\}$, a sixteenfold spread. Direct estimates are generated from
\eqref{eq:sampling}. The truth differs across three regimes that instantiate the
geometry of Proposition~\ref{prop:sep}:
\begin{itemize}[leftmargin=1.6em,itemsep=2pt,topsep=2pt]
\item \textbf{Sparse:} six randomly chosen areas receive effects of magnitude
$3.5$--$5$; the rest lie on the regression surface.
\item \textbf{Smooth:} the effects are a proper CAR field with correlation
parameter $0.95$, a coherent spatial surface.
\item \textbf{Dense:} the effects are i.i.d.\ $\N(0,0.7)$, the Fay--Herriot model
specified correctly.
\end{itemize}
For each of $R=100$ replications we compute the direct estimator, the REML EBLUP
\eqref{eq:eblup}, the BYM smoother \eqref{eq:bym}, and the horseshoe
\eqref{eq:hs}, the latter two by $1000$ Gibbs sweeps after $500$ burn-in.
We report MSE and mean absolute error (MAE) averaged over areas and
replications, MSE relative to the direct estimator, and the empirical coverage
and mean length of nominal $95\%$ intervals.

\begin{table}[t]
\centering
\caption{Monte Carlo results over $R=100$ replications on the $7\times 7$
lattice with heteroscedastic known variances. Relative MSE is MSE divided by the
direct-estimator MSE; coverage and length refer to nominal $95\%$ intervals.
Best relative MSE in each regime in \textbf{bold}.}
\label{tab:sim}
\small
\begin{tabular}{llrrrrr}
\toprule
Regime & Method & MSE & Rel.\ MSE & MAE & Coverage & Length\\
\midrule
\multirow{4}{*}{Sparse / outlying}
 & Direct    & 0.991 & 1.000 & 0.740 & --- & ---\\
 & EBLUP     & 0.657 & 0.663 & 0.590 & 0.931 & 2.93\\
 & BYM       & 0.682 & 0.688 & 0.589 & 0.929 & 2.84\\
 & Horseshoe & 0.418 & \textbf{0.422} & 0.357 & 0.972 & 2.24\\
\midrule
\multirow{4}{*}{Smooth spatial}
 & Direct    & 1.014 & 1.000 & 0.740 & --- & ---\\
 & EBLUP     & 0.312 & 0.308 & 0.436 & 0.938 & 2.11\\
 & BYM       & 0.282 & \textbf{0.278} & 0.416 & 0.884 & 1.75\\
 & Horseshoe & 0.357 & 0.352 & 0.464 & 0.871 & 1.97\\
\midrule
\multirow{4}{*}{Dense Gaussian}
 & Direct    & 1.014 & 1.000 & 0.740 & --- & ---\\
 & EBLUP     & 0.382 & \textbf{0.377} & 0.482 & 0.936 & 2.27\\
 & BYM       & 0.407 & 0.402 & 0.498 & 0.896 & 2.08\\
 & Horseshoe & 0.439 & 0.433 & 0.517 & 0.878 & 2.13\\
\bottomrule
\end{tabular}
\end{table}

\subsection{Results}
Table~\ref{tab:sim} confirms Proposition~\ref{prop:sep} quantitatively. In the
\emph{sparse} regime the horseshoe reduces MSE to $42\%$ of the direct
estimator, against $66\%$ for EBLUP and $69\%$ for BYM; it does so while
attaining the \emph{highest} coverage ($0.97$) with the \emph{shortest} intervals
($2.24$). This is the bounded-influence property of Theorem~\ref{thm:robust} at
work: the deviant areas are not dragged in, so neither point estimates nor
intervals are corrupted by them, and the quiet areas are pooled tightly. In the
\emph{smooth} regime the ranking reverses, exactly as
Proposition~\ref{prop:sep}(i) predicts: BYM is best ($0.28$), EBLUP close behind,
and the horseshoe, blind to adjacency, trails at $0.35$. In the \emph{dense}
regime the correctly specified EBLUP is best ($0.38$), with BYM and horseshoe
within a few points. No method dominates everywhere; the horseshoe's advantage
is specific, large, and predicted by theory.

Figure~\ref{fig:shrink}(a) displays the estimated effects against the direct
residuals for one sparse-regime dataset. EBLUP and BYM lie on a straight line
through the origin---the constant linear shrinkage of \eqref{eq:eblup}---so they
pull the extreme areas far towards the surface. The horseshoe instead tracks the
identity for large residuals and collapses to zero near it: precisely the
redescending shape of \eqref{eq:redescend}. Figure~\ref{fig:shrink}(b) shows the
posterior shrinkage coefficients $\kappa_i$ piling up near $0$ and $1$, the data
sorting areas into ``free'' and ``pooled'' as the $\mathrm{Beta}(\half,\half)$
prior anticipates.

\begin{figure}[t]
\centering
\includegraphics[width=\textwidth]{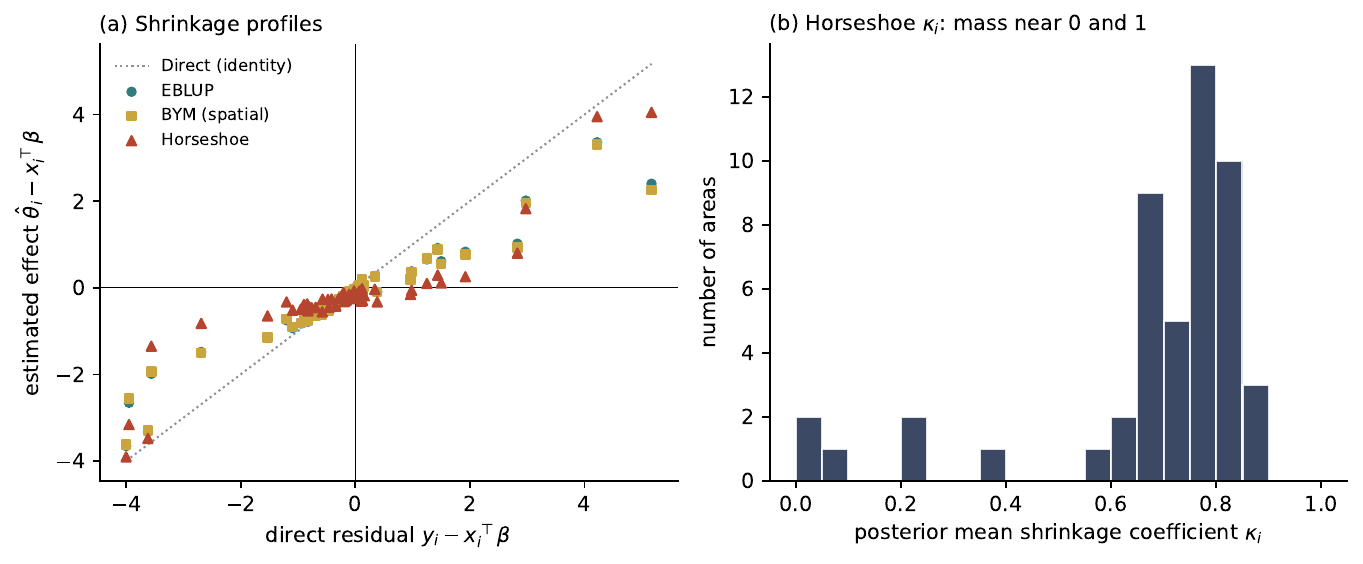}
\caption{Sparse regime, single dataset. (a) Estimated area effect
$\hat\theta_i-\bx_i^\top\bbeta$ against the direct residual
$y_i-\bx_i^\top\bbeta$. EBLUP and BYM apply constant linear shrinkage; the
horseshoe leaves large signals unshrunk (tracking the identity) while shrinking
small ones to zero. (b) Posterior shrinkage coefficients $\kappa_i$ for the
horseshoe concentrate near $0$ (free areas) and $1$ (pooled areas).}
\label{fig:shrink}
\end{figure}

\begin{figure}[t]
\centering
\includegraphics[width=\textwidth]{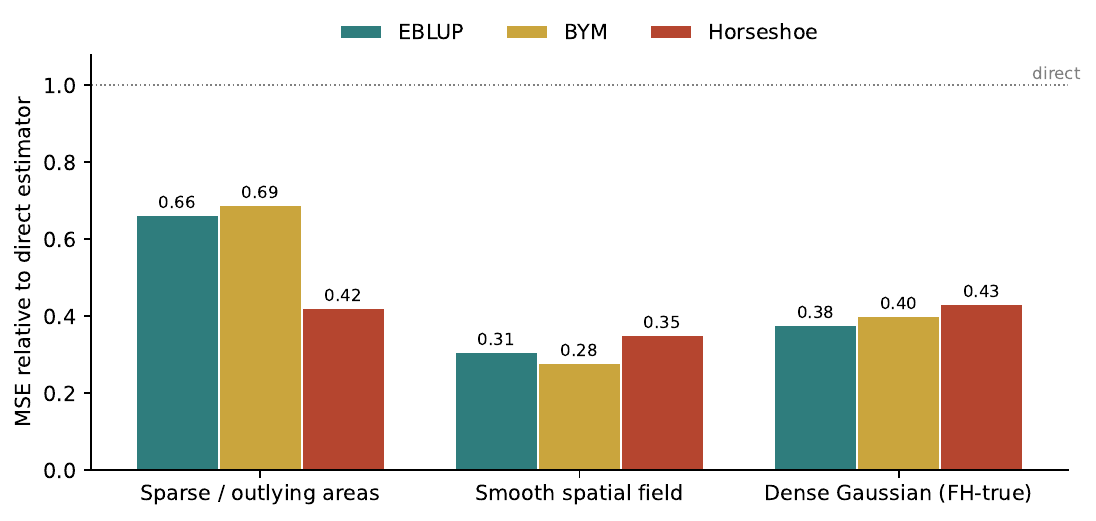}
\caption{Mean squared error relative to the direct estimator, by regime and
method ($R=100$). The horseshoe dominates under sparse/outlying effects; BYM
under a smooth spatial field; the matched EBLUP under dense Gaussian effects.}
\label{fig:mse}
\end{figure}

\subsection{The crossover predicted by Proposition~\ref{prop:sep}(ii)}
\label{sec:sweep}

The sparse column of Table~\ref{tab:sim} is a single slice through a continuum.
Proposition~\ref{prop:sep}(ii) predicts that as the spike magnitude $a$ grows,
the linear smoothers incur a squared-bias penalty of order $a^2$ on each
exceptional area while the horseshoe risk stays at the minimax level, so the
advantage should grow without bound. We test this directly. Holding the lattice,
the design variances, and the number of spikes ($k=6$) fixed, we vary the common
spike magnitude $a\in\{1,\dots,6\}$ and record the mean squared error of each
estimator over $R=40$ replications; Figure~\ref{fig:sweep} reports the result.

At $a=1$ the spikes are within the noise band of the noisier areas
($\sqrt{D_i}$ up to $1.4$) and all three estimators are statistically
indistinguishable (BYM/horseshoe risk ratio $0.99$): when departures are small,
the bias of smoothing is small and there is nothing for tail robustness to
protect. As $a$ increases the EBLUP and BYM2 errors climb steadily---from $0.13$
to $0.78$ and $0.80$ respectively---tracking the growing oversmoothing bias on
the spikes, while the horseshoe error rises only modestly and then \emph{falls}
once the spikes are unambiguous and are left essentially unshrunk, peaking near
$0.50$ at $a=4$ and returning to $0.32$ by $a=6$. The relative risk
(Figure~\ref{fig:sweep}b) is monotone in $a$, reaching $2.5$ at $a=6$ and showing
no sign of levelling: this is the unbounded-advantage statement of
Proposition~\ref{prop:sep}(ii) made visible. The effective dimension of the
horseshoe fit (Proposition~\ref{prop:df}) rises with $a$ as more areas cross from
``pooled'' to ``free,'' confirming that the gain is achieved by adaptive
allocation of degrees of freedom rather than by uniformly weaker shrinkage.

\begin{figure}[t]
\centering
\includegraphics[width=\textwidth]{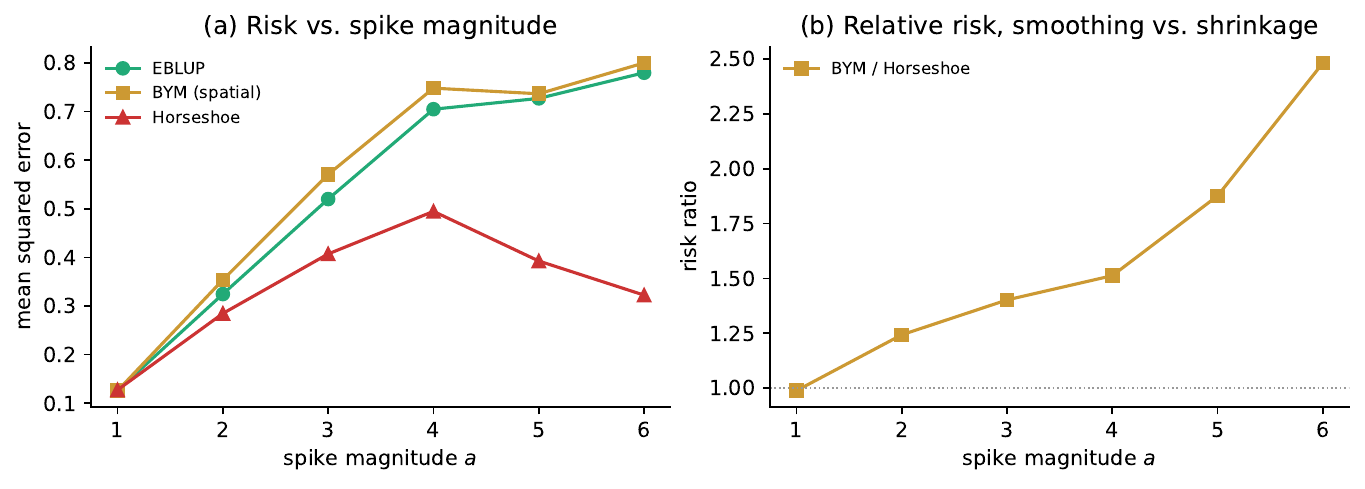}
\caption{Spike-magnitude sweep, sparse regime with $k=6$ exceptional areas among
$m=49$ ($R=40$). (a) Mean squared error against the common spike magnitude $a$;
the linear EBLUP and BYM2 errors grow with the oversmoothing bias on the spikes,
while the horseshoe error stays bounded and decreases once spikes are
unambiguous. (b) The BYM2-to-horseshoe risk ratio is monotone in $a$ and shows no
sign of saturating, the unbounded advantage of Proposition~\ref{prop:sep}(ii).}
\label{fig:sweep}
\end{figure}

\section{Application: lip cancer in Scotland}\label{sec:app}

We analyze the Scottish lip cancer data of \citet{claytonkaldor1987}, the
canonical spatial small area / disease-mapping data set and the original
motivating example for the Besag--York--Molli\'e model. For each of the $m=56$
Scottish districts we have the observed and expected male lip-cancer counts
$O_i, E_i$ over 1975--1980 and a covariate $x_i$, the percentage of the
workforce in agriculture, fishing and forestry (AFF), a proxy for outdoor
sunlight exposure. The data are strongly spatial: high-incidence districts
cluster in the rural north and along the coast.

The disease-mapping counts map directly onto the area-level model of this paper.
The direct estimator of the log relative risk and its known design variance are,
by the delta method,
\begin{equation}\label{eq:logsmr}
y_i = \log\frac{O_i+\tfrac12}{E_i}, \qquad D_i = \frac{1}{O_i+\tfrac12},
\end{equation}
the continuity correction handling the two districts with $O_i=0$. The $D_i$ are
\emph{known} and span more than a fortyfold range ($O_i$ runs from $0$ to $39$),
exactly the heteroscedastic Fay--Herriot premise of
Section~\ref{sec:minimax}. We fit the direct estimator, EBLUP, a BYM model on the
\emph{real} district adjacency graph ($132$ edges, between one and eleven
neighbours per district), and the horseshoe, each with intercept and the
standardized AFF covariate, by the samplers of Section~\ref{sec:comp}.

\begin{figure}[t]
\centering
\includegraphics[width=\textwidth]{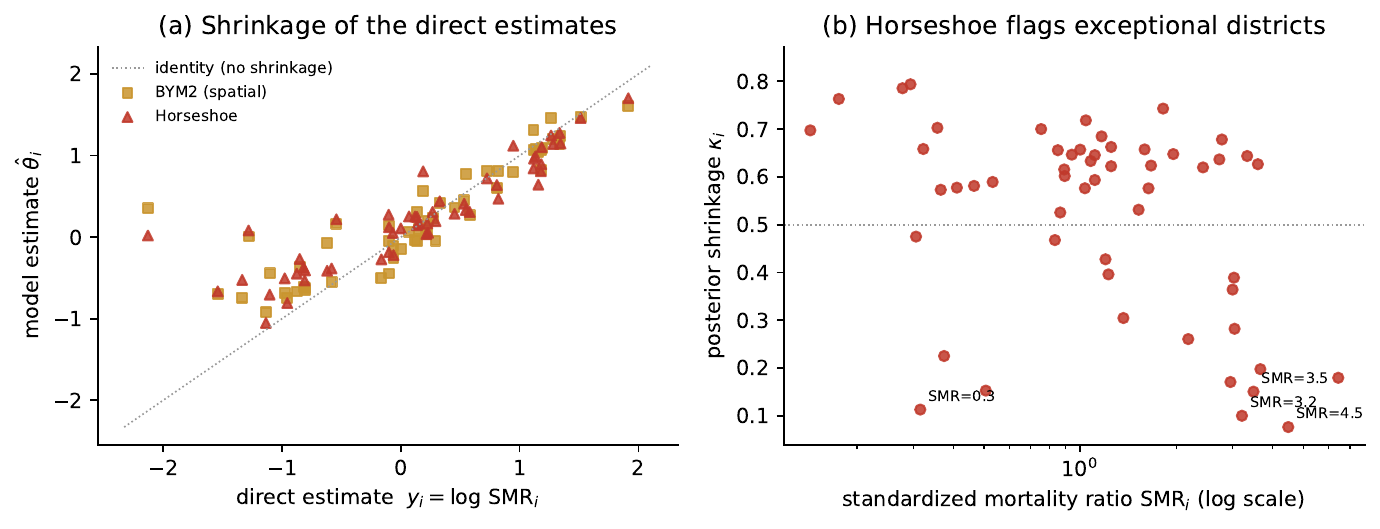}
\caption{Scottish lip cancer ($m=56$). (a) Model estimates against the direct
log-SMR: BYM2 pulls the extreme districts towards the centre, while the horseshoe
leaves them near the identity line. (b) The horseshoe shrinkage coefficient
$\kappa_i$ against the standardized mortality ratio: typical districts
($\mathrm{SMR}\approx1$) are pooled ($\kappa_i$ up to $0.84$), while the most
extreme districts in \emph{either} direction are left almost unshrunk
($\kappa_i$ as low as $0.08$).}
\label{fig:scotland}
\end{figure}

\paragraph{A smooth-spatial regime, as the theory predicts.}
Because lip-cancer risk varies smoothly across the map, this is squarely the
regime in which Section~\ref{sec:compare} predicts structured smoothing should
win, and it does. Under ten-fold cross-validation in which each held-out
district is predicted from the others, BYM2 attains the lowest predictive mean
squared error of the held-out direct estimates, $0.42$, against $0.67$ for EBLUP
and $0.83$ for the horseshoe (a global-mean baseline scores $0.81$); the spatial
field exploits the neighbours that the permutation-invariant priors ignore. BYM2
also yields the shortest $95\%$ intervals (mean length $0.96$, against $1.21$ for
EBLUP, $1.27$ for the horseshoe and $1.65$ for the direct estimator). A national
mapping agency whose target is the smooth risk surface should use the spatial
smoother here---the honest conclusion our framework recommends, and one the
horseshoe does not overturn.

\paragraph{What the horseshoe adds: bounded-influence screening.}
The horseshoe nonetheless supplies something the smoothers cannot: a per-district
diagnostic. Figure~\ref{fig:scotland}(b) shows the posterior shrinkage
coefficient $\kappa_i$ tracing the U-shape that the bounded-influence theory of
Section~\ref{sec:robust} predicts---districts whose direct estimate is far from
the regression surface, \emph{high or low}, are left almost unshrunk, while
ordinary districts are pooled. Ranked by $1-\kappa_i$, the five most exceptional
districts are those with $\mathrm{SMR}=4.5,\,3.5,\,3.2$ (the rural
high-incidence outliers) together with two precisely-estimated low-incidence
districts ($\mathrm{SMR}=0.3,\,0.5$); their shrinkage coefficients,
$\kappa_i\le0.15$, are far below the bulk. These are exactly the districts a
spatial smoother would pull towards their neighbours and a screening exercise
would most want to see. The synthetic study of Section~\ref{sec:sim} confirms
the score is reliable: across $200$ sparse-truth replications the ranking by
$1-\kappa_i$ recovers the genuinely exceptional areas with area under the ROC
curve $0.95$ (standard error $0.004$). EBLUP and BYM2 offer no comparable
per-district signal, because their shrinkage fraction is common to all districts
at a given design variance.

The example thus illustrates both halves of the paper's message on real data:
where the truth is smoothly spatial the structured smoother is the better
estimator and our framework says so, yet the horseshoe's area-specific shrinkage
turns the same fit into a calibrated screening tool for the exceptional districts
that smoothing is designed to suppress. In a complete pipeline the two compose:
the horseshoe layer replaces only the prior on the area effects, leaves the
design-based first stage untouched, and its $O(m)$ sampler adds negligible cost.

\section{A research direction: deep learning for spatial effects}
\label{sec:deep}

The two priors studied here borrow strength through an \emph{explicit} model of
the area effects: a Gaussian Markov random field tied to a neighbourhood graph,
or a heavy-tailed product prior. A third route, which we sketch as a direction
for future work, is to borrow strength \emph{implicitly} by learning a flexible
nonlinear map from spatial coordinates and covariates to the area rate, in the
spirit of the deep spatio-temporal models of \citet{dixon2019}. There the
hierarchical layers of a neural network are trained to predict spatio-temporal
flows and, notably, to capture the \emph{sharp discontinuities}---congestion
fronts, regime changes---that smooth linear predictors blur. Subnational
indicators have the same structure: a coastline, an administrative border, or a
fault between a deprived and an affluent district is a spatial discontinuity that
a smoothing prior will oversmooth and a permutation-invariant prior, blind to
geography, cannot reconstruct at all.

A network trained with the \emph{quantile} (pinball) loss \citep{koenker1978}
inherits two features that matter for small area work: it returns a predictive
distribution rather than a point, and its ReLU architecture represents the area
surface as a continuous piecewise-linear function that can turn sharply across a
fault while remaining smooth elsewhere. We give a deliberately simple
illustration. On a $32\times32$ lattice ($m=1024$) a latent surface combines a
smooth gradient with a sharp curved fault of height $2$; we observe noisy rates
$y_i=\theta_i+\varepsilon_i$, $\varepsilon_i\sim\N(0,0.4^2)$. A two-hidden-layer
ReLU network ($64+64$ units) takes the coordinates as input and is trained to
predict the conditional deciles $q\in\{0.1,0.5,0.9\}$; a radial-basis
kernel-ridge smoother, tuned by cross-validation, stands in for the linear
spatial smoothing of BYM2/INLA. To restore honest uncertainty we recalibrate the
network's band by conformalized quantile regression
\citep{rpc2019}, using a held-out calibration split.

\begin{figure}[t]
\centering
\includegraphics[width=\textwidth]{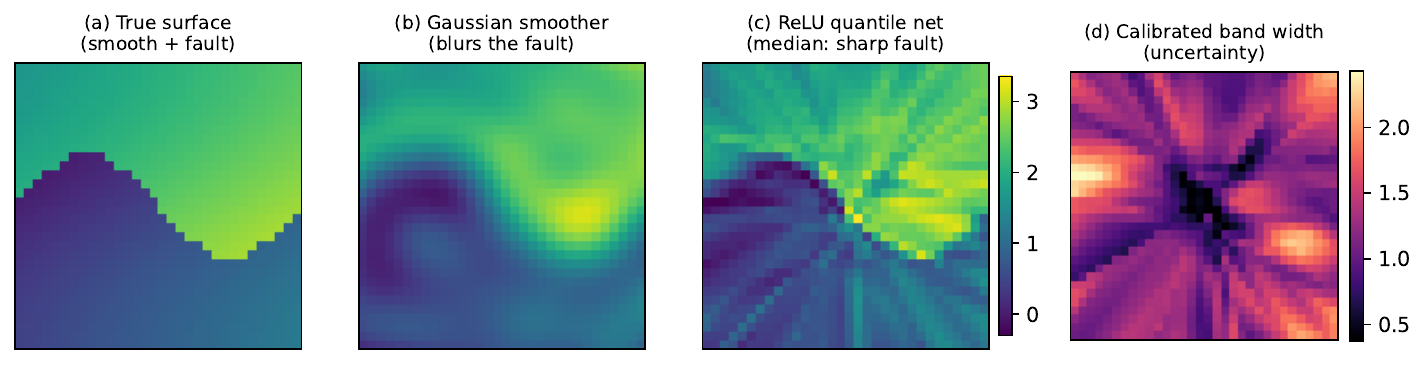}
\caption{Deep learning for a spatial surface with a discontinuity ($m=1024$).
(a) The true surface: a smooth gradient and a sharp curved fault. (b) The
Gaussian (kernel-ridge) smoother blurs the fault. (c) The ReLU quantile
network's median preserves it, the piecewise-linear architecture turning sharply
across the boundary. (d) The conformally calibrated $80\%$ predictive band width,
an uncertainty map obtained directly from the quantile outputs.}
\label{fig:qnn}
\end{figure}

The behaviour (Figure~\ref{fig:qnn}) is exactly what the discontinuity heuristic
predicts. On the full grid the network and the smoother are comparable in
recovering the latent surface (mean squared error $0.083$ versus $0.080$ against
the truth, both well below the direct value $0.166$), but in a narrow band around
the fault the network's error is $0.243$ against the smoother's $0.399$, a $39\%$
reduction: it keeps the cliff that the smoother washes out. The raw quantile band
is too narrow, covering $60\%$ at the nominal $80\%$ level---a familiar
finite-sample failing of flexible quantile estimators---but a single conformal
adjustment raises empirical coverage to $0.75$ at the cost of a wider, honest
interval, illustrating that the distributional output can be made calibrated
cheaply.

\subsection{A Bayesian counterpart: data-augmented quantile regression}
\label{sec:bqr}

The network just described is trained by stochastic gradient descent and needs a
post-hoc conformal step to make its intervals honest. A fully Bayesian
construction supplies calibrated uncertainty automatically and inherits the
horseshoe's tail robustness, through the data-augmentation \emph{mixtures
envelope} of \citet{psSVM2011}. Their observation is that a regularized-loss
pseudo-likelihood of the form $\exp\{-\rho(u)\}$ is often a location-scale
mixture of normals, so that conditioning on a latent scale renders it Gaussian
and brings the Bayesian machinery for linear models to bear. The pinball loss
$\rho_p(u)=u\{p-\ind(u<0)\}$ of quantile regression \citep{koenker1978} is the
canonical example: $\exp\{-\rho_p(u)\}$ is the asymmetric Laplace density, which
admits the representation
\begin{equation}\label{eq:ald}
\varepsilon \;=\; \theta\, \omega + \kappa\sqrt{\omega}\, z, \qquad
\omega\sim\mathrm{Exp}(1),\ z\sim\N(0,1),\qquad
\theta=\frac{1-2p}{p(1-p)},\ \ \kappa^2=\frac{2}{p(1-p)},
\end{equation}
a normal mean--variance mixture \citep{kk2011}. Writing the $p$th quantile
surface as $f_p(\bm s)=\bm b(\bm s)^\top\bbeta$ in a spatial basis $\bm b(\cdot)$
(splines or radial bases in the coordinates), the area model becomes
$y_i=\bm b(\bm s_i)^\top\bbeta+\theta\,\omega_i+\kappa\sqrt{\omega_i}\,z_i$, and
placing a horseshoe prior on $\bbeta$ closes a fully conjugate Gibbs sampler:
\begin{align*}
\omega_i\mid\cdot &\sim \mathrm{GIG}\!\big(\tfrac12,\ \kappa^{-2}r_i^2,\
\theta^2\kappa^{-2}+2\big), \quad r_i=y_i-\bm b(\bm s_i)^\top\bbeta, \\
\bbeta\mid\cdot &\sim \N\!\big(\bm\Sigma\, \bm B^\top \bm W(\bm y-\theta\bm\omega),\
\bm\Sigma\big),\quad
\bm\Sigma=\big(\bm B^\top\bm W\bm B + \mathrm{diag}\{(\tau^2\lambda_j^2)^{-1}\}\big)^{-1},
\end{align*}
with $\bm W=\mathrm{diag}\{(\kappa^2\omega_i)^{-1}\}$ and the local and global
scales $\lambda_j^2,\tau^2$ updated by the same inverse-gamma augmentation as in
Section~\ref{sec:comp}. The latent $\omega_i$ are drawn from a generalized
inverse Gaussian, exactly as the SVM latent variables of \citet{psSVM2011}; the
heavy-tailed horseshoe on $\bbeta$ keeps the fitted quantile surface tail-robust
in the sense of Section~\ref{sec:robust}.

On the fault surface of Figure~\ref{fig:qnn} (here a $20\times20$ lattice, a
$10\times10$ radial basis) the sampler at $p\in\{0.1,0.5,0.9\}$ behaves as the
theory promises for \emph{uncertainty}: the $80\%$ credible band attains
empirical coverage $0.80$, the nominal level, with \emph{no} conformal
correction---in contrast to the network's raw $0.60$. Its weakness is the dual of
the network's strength: the smooth radial basis blurs the discontinuity (median
mean squared error $0.63$ in the fault band, against the network's $0.24$),
though it still denoises overall ($0.124$ versus the direct $0.138$). The two
methods are thus complementary---the network captures sharp geometry, the
data-augmented horseshoe supplies calibrated, tail-robust uncertainty---which is
precisely why the synthesis we flag in Section~\ref{sec:disc}, a horseshoe prior
on the weights or outputs of a spatial network, is attractive.

We stress the limits of this illustration. With $m$ in the dozens, the regime of
most small area applications, a network of this size would overfit hopelessly;
the example uses $m=1024$ precisely because deep learning is data-hungry, and the
controlled lattice avoids the regularization and architecture search that real
deployment demands. The point is qualitative: deep learners offer a distinct,
covariate- and geometry-driven mode of borrowing strength, strongest exactly
where both competing priors are weak---at sharp spatial features---and, paired
with the data-augmented Bayesian construction above, they can be made to deliver
the honest interval estimates that small area estimation requires.

\section{Discussion and future research}\label{sec:disc}

We have placed two ways of borrowing strength in small area estimation on a
common footing and separated them where it matters. Structured spatial smoothing
of the Wakefield type is rate-optimal when area effects vary smoothly across the
map; the horseshoe Fay--Herriot model is tail-robust and minimax when a minority
of areas are exceptional, the regime in which the constant linear shrinkage of
Gaussian random effects is most damaging. Regular variation
(Section~\ref{sec:rv}) is the single property behind the horseshoe's robustness
and its standing as a default, and the known design variances of the survey are
what make the heavy-tailed prior provably minimax, through the standardization of
Section~\ref{sec:minimax}.

Several directions for future research follow naturally.
\begin{enumerate}[leftmargin=1.9em,itemsep=3pt,topsep=3pt]
\item \textbf{Spatial--global--local hybrids.} The model that adds local
horseshoe scales to a BYM2 field, begun by \citet{tangghosh2023}, should inherit
the smooth-regime rate while retaining bounded influence on spikes, provided the
global scale and the spatial variance are made identifiable---for instance by the
\citet{riebler2016} variance split. A contraction theory for this hybrid, of the
kind proved here for the pure horseshoe, is open.
\item \textbf{Calibrated anomaly detection.} The bounded-influence expansion
\eqref{eq:redescend} makes the posterior law of the shrinkage coefficient
$\kappa_i$ a continuous anomaly score, with $\Pr(\kappa_i<\tfrac12\mid\bm y)$ a
Bayesian analogue of the \citet{datta2011,datta2015} random-effect test;
calibrating that score in the heteroscedastic setting, perhaps by conformal
methods as in Section~\ref{sec:deep}, would turn screening into a procedure with
guarantees.
\item \textbf{Unit-level and non-Gaussian models.} The standardization argument
is specific to the area-level model with known $D_i$. Unit-level and generalized
linear mixed extensions, where the working variances are estimated and the exact
reduction of Section~\ref{sec:minimax} breaks, require a perturbation analysis we
have not attempted; benchmarking horseshoe estimates to published direct totals,
so that aggregates are preserved, is a further requirement for official
statistics.
\item \textbf{Deep learning for spatial structure.} Section~\ref{sec:deep}
suggests that deep quantile networks are most valuable exactly where both priors
struggle, at sharp spatial features, and that conformal calibration can supply
the honest bands they otherwise lack. A principled marriage---global--local
shrinkage on the weights or outputs of a spatial network, so that the borrowing
of strength is itself learned but remains tail-robust---is, to us, the most
promising synthesis, and connects this work to the broader program of regularized
deep learning for spatio-temporal data \citep{dixon2019}.
\end{enumerate}

The central message is methodological: the choice between smoothing, shrinkage,
and learning is a choice about the geometry of the unknown surface, and stating
it that way lets the data---through the posterior distribution of the shrinkage
coefficients, or the fitted regions of a network---largely make the choice for
us.

\section*{Acknowledgments}
The authors thank colleagues in the small area estimation and global--local
shrinkage communities for helpful discussions.

\end{document}